\documentclass[12pt,a4paper]{article}
\usepackage{fullpage}
\usepackage{amssymb,amsmath,stmaryrd,amsthm}
\usepackage[mathscr]{eucal}
\usepackage{pstricks,pst-node,pst-text,pst-3d}

\theoremstyle{definition}

\newtheorem{theorem}{Theorem}

\newtheorem{lemma}{Lemma}
\newtheorem{corollary}{Corollary}

\newtheorem{example}{Example}
\theoremstyle{remark}
\newtheorem{remark}{Remark}

\def\eC{\mathscr{C}}

\def\cC{\mathcal{C}}
\def\cF{\mathcal{F}}
\def\cJ{\mathcal{J}}
\def\cN{\mathcal{N}}
\def\cO{\mathcal{O}}
\def\cW{\mathcal{W}}
\def\bR{{\mathbb R}}
\def\tC{\widetilde{\mathcal{C}}}
\def\tF{\widetilde{\mathcal{F}}}
\def\tW{\widetilde{\mathcal{W}}}

\begin{document}

\title{Ensuring the boundedness of the core of games with restricted cooperation}
\author{Michel GRABISCH\\
Paris School of Economics, University of Paris I\\
106-112, Bd de l'H\^opital, 75013 Paris, France\\
Tel (+33) 1-44-07-82-85, Fax
    (+33) 1-44-07-83-01\\
    email \texttt{michel.grabisch@univ-paris1.fr}}
\date{}
\maketitle

\begin{abstract}
The core of a cooperative game on a set of players $N$ is one of the most
popular concept of solution. When cooperation is restricted (feasible coalitions
form a subcollection $\cF$ of $2^N$), the core may become unbounded, which makes
it usage questionable in practice. Our proposal is to make the core
bounded by turning some of the inequalities defining the core into equalities
(additional efficiency constraints). We address the following mathematical
problem: can we find a minimal set of inequalities in the core such that, if
turned into equalities, the core becomes bounded? The new core obtained is
called the restricted core. We completely solve the question when $\cF$ is a
distributive lattice, introducing also the notion of restricted Weber set. We
show that the case of regular set systems amounts more or less to the case of
distributive lattices. We also study the case of weakly union-closed systems and
give some results for the general case.
\end{abstract}

\section{Introduction}
In cooperative models, one of the main issues is to define in a rational way the
sharing of the total worth of a game among the players, what is usually called
the \emph{solution} of the game. The core is perhaps the most popular concept of
solution, because it is built on a very simple rationality criterion: no
coalition should receive less than that it can earn by itself, thus avoiding any
instability in the game (this is often called \emph{coalitional
  rationality}). The core  is a bounded convex polyhedron whenever nonempty, and its
properties have been studied in depth (see, e.g., \cite{sha71,ich81,nura98}).

The classical setting of cooperative games stipulates that any player can
(fully) participate or not participate to the game, and that any coalition can
form. This too simplistic framework has been made more flexible in many
respects, or more tailored to some special kind of application by many authors:
let us cite on the one hand multichoice games \cite{hsra93,peza05}, games with
multiple alternatives \cite{amcama98,bol00} and bicooperative games
\cite{bifejilo06,lagr06a} (participation is gradual, can be positive or
negative, or the player has several options), and on the other hand, games with
restricted cooperation, where only a limited set of coalitions are allowed to
form. A vast literature is devoted to this last category, studying
various possibilities for the algebraic structure of the set of \emph{feasible}
coalitions: games on antimatroids \cite{albibrji04}, convex geometries
\cite{bileji98}, lattices \cite{dere98,fake92,grla05a}, graphs
\cite{subohaquko05,brlava07}, etc.

Our study will concern games with restricted cooperation, and especially the
core of such games. Here also, there is a vast literature we will
not cite here (see a recent survey by the author on this topic \cite{gra09a}). Indeed,
the study of the core in such a general situation becomes much more challenging:
since the core is defined by a system of linear inequalities, it is always a
polyhedron, however it need not be bounded any more, and it may
even have no vertice or it may contain a line. As a matter of fact, since the
core is supposed to represent a set of payoffs for players, boundedness is
perhaps the property one wants to keep in priority (arbitrarily large payoffs
cannot exist in reality). Therefore, a natural
question arises: \emph{How to make to core bounded in any case, keeping the
  spirit of its definition?} By ``spirit of definition'', we mean the essential
idea of coalitional rationality. A very simple answer to this
question was proposed by Grabisch and Xie \cite{grxi05,grxi07}: turn some of the
inequalities into equalities, which can be seen as adding supplementary
\emph{efficiency} constraints, while preserving coalitional rationality. The authors proposed a systematic way of doing
this for games on distributive lattices, according to some interpretation related
to the hierarchy of players. 

We want to take here a more general and mathematical point of
view. Specifically, we address the following question: \emph{Suppose $v$ is a
  game with restricted cooperation, whatever the structure of its set of feasible
coalitions. Can we find a minimal set of inequalities in the core of $v$ such
that, if turned into equalities, the core will be bounded?} A second question
is: what about the Weber set? Can we define it so that the classical property of
inclusion of the core into the Weber set is preserved?

We give a complete answer to these questions for games on distributive lattices,
---thus generalizing and simplifying results of Grabisch and Xie,
and partial answers for other structures and the general case. 

The paper is organized as follows. Section~\ref{sec:prel} introduces the basic
material for the paper: set systems, posets and lattices, etc. We also explain
our main idea to make the core bounded. Section~\ref{sec:cloui} studies the
case of distributive lattices. It gives an optimal algorithm to find which
inequalities must be turned into equalities. Also, it introduces the notion of
restricted Weber set, and shows that the classical result of inclusion of the
core into the Weber set still holds. Section~\ref{sec:gene} studies the general
case. A first result shows that if rays have a certain form, one can treat an
equivalent problem where the set system is a distributive lattice, and therefore
benefit from results of Section~\ref{sec:cloui}. It is shown that regular set
systems fall into this category. However, for weakly union-closed systems, an
additional condition on the set system is required. We give also an algorithm to
find all extremal rays of the core of a game on a regular set system.

We assume some familiarity of the reader with polyhedra.
To avoid a heavy notation, we often omit braces and commas for singletons and sets, writing e.g,
$N\setminus i$ instead of $N\setminus \{i\}$, $123$ instead of $\{1,2,3\}$, etc. 


\section{Preliminaries}\label{sec:prel}
\subsection{Games on set systems}\label{sec:ss}
We consider $N:=\{1,\ldots,n\}$ the set of players, agents, etc. A \emph{set
  system} $\cF$ on $N$ is a collection of subsets of $N$ containing $N$ and
$\emptyset$. One can think of $\cF$ as the collection of \emph{feasible
  coalitions}, and when $\cF\subset 2^N$ it is common to speak of \emph{restricted
  cooperation}. A \emph{game} on $\cF$ is a function $v:\cF\rightarrow \bR$ such
that $v(\emptyset)=0$.

A \emph{payoff vector} $x$ is any vector in $\bR^N$, which defines the
amount of money given to each player. It is common to use the notation $x(S)$ where $S\in
2^N$, as a shorthand for $\sum_{i\in S}x_i$, with the convention
$x(\emptyset):=0$.   
The \emph{core} of a game $v$ on $\cF$ is the set of payoff vectors being
\emph{coalitionally rational}, in the sense that any feasible coalition $S$ gets
at least what it can achieve by itself, namely $v(S)$: 
\[
\cC(v):=\{x\in \bR^N\mid x(S)\geq v(S),\forall S\in\cF,x(N)=v(N)\}.
\]
The equality $x(N)=v(N)$ is known as the \emph{efficiency condition}. It means
that no more than $v(N)$ can be distributed among the players whatsoever, and
distributing strictly less would be inefficient (the definition makes sense only
if the grand coalition $N$ is the best way to make profit).  

By definition, the core is a closed convex polyhedron, however it may be
unbounded (see in \cite{gra09a} a survey of the properties of the core of games
on set systems). We denote by $\cC(0)$ the \emph{recession cone} of $\cC(v)$,
that is, the cone defined by
\[
\cC(0):=\{x\in \bR^N\mid x(S)\geq 0,\forall S\in\cF,x(N)=0\}.
\]
It is well known from the theory of polyhedra that $\cC(v)$ is bounded if and
only if $\cC(0)=\{0\}$, and that the extremal rays of $\cC(0)$ are the extremal rays of $\cC(v)$
for any game $v$. Since in this paper we are mainly interested into the
boundedness issue and therefore in rays, we mainly deal with the recession cone
$\cC(0)$. 

\subsection{Posets and lattices}\label{sec:latt}
A set system $\cF$ can be seen as a partially ordered set (poset) when endowed
with the inclusion order. Properties of $\cC(v)$ substantially differ according
to the algebraic structure of $(\cF,\subseteq)$. We give here some fundamental
notions on posets which will be useful in the sequel (see,
e.g., Davey and Priestley \cite{dapr90} for details).

A \emph{partially ordered set} $(P,\leq)$, or \emph{poset} for short, is a set
$P$ endowed with a partial order $\leq$ (reflexive, antisymmetric and
transitive). As usual, the asymmetric part of $\leq$ is denoted by $<$. For
$x,y\in (P,\leq)$ (if no ambiguity occurs, we may write simply $P$), we write
$x\prec y$ and say that $x$ is \emph{covered} by $y$ if $x<y$ and there is no
$z\in P$ such that $x<z<y$. An element $x\in P$ is minimal if there is no $y\in
P$ such that $y<x$.

A \emph{chain} from $x$ to $y$ in $P$ is any sequence $x,x_1,\ldots,x_p,y$ of
elements of $P$ such that $x<x_1<\cdots< x_p<y$. The chain is \emph{maximal} if
no other chain from $x$ to $y$ contains it, i.e., if $x\prec
x_1\prec\cdots\prec x_p\prec y$. The \emph{length} of a chain is its number of elements
minus 1.

The \emph{height} of $x\in P$, denoted by $h(x)$, is the length of a longest
chain from a minimal element to $x$. Elements of
same height $l$ form \emph{level} $l+1$. Hence, level 1 (denoted by $L_1$) is the set
of all minimal elements, level 2 (denoted by $L_2$) is the set of minimal
elements of $P\setminus L_1$, etc. The height of $P$, denoted by $h(P$), is
the maximum of $h(x)$ taken over all elements of $P$.

Consider a poset $(P,\leq)$ and some $Q\subseteq P$. Then $Q$ is a
\emph{downset} of $P$ if $x\in Q$ and $y\leq x$ imply $y\in Q$. Any element
$x\in P$ generates a downset, defined by $\downarrow\! x:=\{y\in P\mid y\leq
x\}$. We denote by $\cO(P)$ the set of all downsets of $P$.

\medskip

A lattice $(L,\leq)$ is a poset having the following property: for any $x,y\in
L$, their supremum and infimum, denoted by $x\vee y$ and $x\wedge y$, exist in
$L$. When a lattice is finite, it has a greatest element $\top=\bigvee_{x\in
  L}x$ (top element), and a smallest element $\bot=\bigwedge_{x\in L}x$ (bottom
element). If $\vee,\wedge$ obey distributivity, then $L$ is said to be
\emph{distributive}. In a distributive lattice $L$, all
maximal chains from $\bot$ to $\top$ have same length $h(L)$.

Given a lattice $L$, an element $x\in L$, $x\neq\bot$, is
\emph{join-irreducible} if it cannot be expressed as the supremum of other
elements, or equivalently, if it covers a unique element. We denote by $\cJ(L)$
the set of join-irreducible elements of $L$. It can be shown that if $L$ is
distributive, its height $h(L)$ equals the number of join-irreducible elements
$|\cJ(L)|$. 

\medskip

Finite distributive lattices have a remarkable property: they are completely determined by
their join-irreducible elements. Specifically, consider $(L,\leq)$ a
distributive lattice, and $(\cJ(L),\leq)$ its join-irreducible elements
considered as a subposet of $L$. Then Birkhoff's theorem \cite{bir67} says that
$(L,\leq)$ and $(\cO(\cJ(L)),\subseteq)$ are isomorphic. Conversely, any poset
$(P,\leq)$ \emph{generates} a distributive lattice $(\cO(P),\subseteq)$ (hence, we deduce a
characterization of distributive lattices: the set of distributive lattices of
height $n$ is in bijection with the set of posets of $n$
elements). Figure~\ref{fig:birk} illustrates this fundamental result. 
\begin{figure}[htb]
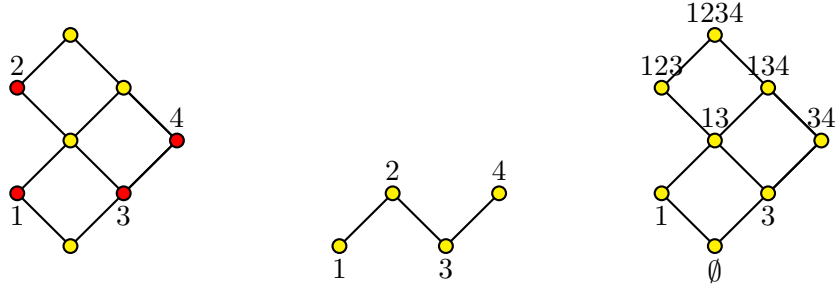

\begin{center}
\psset{unit=0.7cm}
\pspicture(0,0)(3,4)
\pspolygon(1,0)(3,2)(2,3)(0,1)
\pspolygon(2,1)(3,2)(1,4)(0,3)
\pscircle[fillstyle=solid,fillcolor=yellow](1,0){0.15}
\pscircle[fillstyle=solid,fillcolor=red](0,1){0.15}
\pscircle[fillstyle=solid,fillcolor=yellow](1,2){0.15}
\pscircle[fillstyle=solid,fillcolor=red](2,1){0.15}
\pscircle[fillstyle=solid,fillcolor=red](3,2){0.15}
\pscircle[fillstyle=solid,fillcolor=red](0,3){0.15}
\pscircle[fillstyle=solid,fillcolor=yellow](2,3){0.15}
\pscircle[fillstyle=solid,fillcolor=yellow](1,4){0.15}
\uput[-90](0,1){\small $1$}
\uput[-90](2,1){\small $3$}
\uput[90](3,2){\small $4$}
\uput[90](0,3){\small $2$}
\endpspicture
\hspace*{2cm}
\psset{unit=0.7cm}
\pspicture(0,0)(3,4)
\psline(0,0)(1,1)(2,0)(3,1)
\pscircle[fillstyle=solid,fillcolor=yellow](0,0){0.15}
\pscircle[fillstyle=solid,fillcolor=yellow](1,1){0.15}
\pscircle[fillstyle=solid,fillcolor=yellow](2,0){0.15}
\pscircle[fillstyle=solid,fillcolor=yellow](3,1){0.15}
\uput[-90](0,0){\small $1$}
\uput[90](1,1){\small $2$}
\uput[-90](2,0){\small $3$}
\uput[90](3,1){\small $4$}
\endpspicture
\hspace*{2cm}
\psset{unit=0.7cm}
\pspicture(0,0)(3,4)
\pspolygon(1,0)(3,2)(2,3)(0,1)
\pspolygon(2,1)(3,2)(1,4)(0,3)
\pscircle[fillstyle=solid,fillcolor=yellow](1,0){0.15}
\pscircle[fillstyle=solid,fillcolor=yellow](0,1){0.15}
\pscircle[fillstyle=solid,fillcolor=yellow](1,2){0.15}
\pscircle[fillstyle=solid,fillcolor=yellow](2,1){0.15}
\pscircle[fillstyle=solid,fillcolor=yellow](3,2){0.15}
\pscircle[fillstyle=solid,fillcolor=yellow](0,3){0.15}
\pscircle[fillstyle=solid,fillcolor=yellow](2,3){0.15}
\pscircle[fillstyle=solid,fillcolor=yellow](1,4){0.15}
\uput[-90](1,0){\small $\emptyset$}
\uput[-90](0,1){\small $1$}
\uput[-90](2,1){\small $3$}
\uput[90](0,3){\small $123$}
\uput[90](3,2){\small $34$}
\uput[90](1,2){\small $13$}
\uput[90](2,3){\small $134$}
\uput[90](1,4){\small $1234$}
\endpspicture
\end{center}
\caption{Left: a distributive lattice $L$. Join-irreducible elements are those in
  dark grey. Center: the poset $\cJ(L)$ of join-irreducible elements. Right: the set
  $\cO(\cJ(L))$ of all downsets of $\cJ(L)$ ordered by inclusion, which is
  isomorphic to $L$}
\label{fig:birk}
\end{figure}

\subsection{Main families of set systems}\label{sec:fam}
Among the numerous families of sets systems, we put emphasis on three of them:
distributive lattices of height $n$, regular set systems, and weakly union-closed systems.

$(\cF,\subseteq)$ is a distributive lattice is equivalent to say that $\cF$ is a
set system closed under union and intersection, and by Birkhoff's result, it is
generated by a poset of $n$ elements if and only if $\cF$ has height $n$. This poset can be
interpreted as the set of players $N$ endowed with some partial order $\leq$,
which can be thought of as a hierarchy on players or a precedence order (and
then we recover exactly games with precedence constraints of Faigle and Kern
\cite{fake92}): see again Figure~\ref{fig:birk}, center (the hierarchy) and
right (the set system).
\begin{remark}
We discard from the analysis distributive lattices $L$ of height smaller than $n$:
essentially, it amounts to redefine the set of players as $N'$ with $|N'|=h(L)$,
where some of the players of $N$ have been regrouped into ``macro-players''.  
\end{remark} 

A set system is \emph{regular} if all its maximal
chains from $\emptyset$ to $N$ are of length $n$ (see
\cite{hogr05,hogr05a,lagr06b} for works dealing with regular set systems).
Evidently, any set system closed under union and intersection of height $n$ is
regular, but the converse is not true. 

A set system $\cF$ is \emph{weakly union-closed} if for any $S_1,S_2\in \cF$
such that $S_1\cap S_2\neq\emptyset$, we have $S_1\cup S_2\in \cF$ (see \cite{fagrhe09}
and also \cite{albibolo01} where weakly union-closed systems are called union stable
structures).

These three families are distinct, and no family is included into another one
(see \cite{gra09a}).

\subsection{How to make the core bounded}\label{sec:rest}
Given a set system $\cF$, our main goal is to modify the definition of the core
to make it bounded for any game $v$, by replacing some of the inequalities by
equalities (evidently, the core will become bounded after a sufficient number of
such operations). Observe that doing so preserves the coalitional rationality
principle, and this can be interpreted as adding new efficiency constraints.

We call \emph{normal sets} the sets $S\in\cF$ corresponding to inequalities
$x(S)\geq v(S)$ turned into equalities $x(S)=v(S)$, provided the collection of
those sets makes the core bounded (recall that this is independent of $v$ since
it suffices to study the recession cone $\cC(0)$). Such a collection (called a \emph{normal
  collection}), is denoted by $\cN:=\{N_1,\ldots,N_q\}$, and we make the
convention that $N$ is \emph{not} an element of $\cN$. We call the core with
these additional equalities the \emph{core restricted by the normal collection
  $\cN$}, or if no ambiguity occurs, the \emph{restricted core}, and denote it
by $\cC_\cN(v)$.

As mentionned in the introduction, Grabisch and Xie have proposed a particular
way to define a normal collection when $\cF$ is a distributive lattice. Suppose
$\cF$ is a distributive lattice of height $n$, with generating poset
$(N,\leq)$. As mentionned in Section~\ref{sec:latt}, the height function on $(N,\leq)$
induces a partition of $N$ into levels $L_1,\ldots,L_q$. Then the normal
collection of Grabisch and Xie is simply $(L_1,L_1\cup L_2,\ldots,L_1\cup\cdots\cup
L_{q-1})$. Note that the obtained normal collection is \emph{nested}, i.e., it
forms a chain in $\cF$.

\section{Case of distributive lattices of height $n$}\label{sec:cloui}
\subsection{Normal sets}\label{sec:normal}
We know from the previous section that these set systems are closed under union
and intersection, that they possess $n$ join-irreducible elements, and that they
are generated by a poset $(N,\leq)$ (i.e., $\cF=\cO(N,\leq)$). We recall that
$i\prec j$ means that $i<j$ and there is no $k\in N$ such that $i<k<j$.

For those sets systems, we know the following result from Tomizawa
\cite{tom83}. We denote by $J_i$, $i\in N$, the join-irreducible element of
$\cF$ induced by $i$, that is simply, $J_i=\downarrow\!i$. 
\begin{theorem}\label{th:1}
The extremal rays of $\cC(0)$ are of the form $(1_j,-1_i)$, with $i\in N$ such
that $|J_i|>1$, $j\in J_i$ and $j\prec i$.
\end{theorem}
Here we use the notation $1_i$ for the vector of $\bR^N$ having component $i$ equals to 1
and 0 otherwise, and similarly for $(1_j,-1_i)$, etc.

Recall that $\cC(v)$ will become bounded if there is no more
extremal rays in $\cC(0)$. Therefore, we must
study how inequalities turned into equalities can ``kill'' extremal rays of $\cC(0)$. 

 The following can be proved.
\begin{lemma}\label{lem:1}
Consider $J_i$, $|J_i|>1$, $j\prec i$. The extremal ray $(1_j,-1_i)$ is killed
by equality $x(F)=0$ if and only if $j\in F$ and $i\not\in F$.
\end{lemma}
\begin{proof}
($\Leftarrow$) Suppose that $j\in F$ and $i\not\in F$. Then, if $x\in \cC(0)$
  satisfies $x(F)=0$, we have
\begin{equation}\label{eq:1}
x(F) = x_j+\sum_{k\in F\setminus j}x_k= 0. 
\end{equation}
Consider now $x':=x+\alpha(1_j,-1_i)$, $\alpha>0$. Then $x'$ does not satisfy
equality $x'(F)=0$ since
\[
x'(F) = x_j+\alpha + \sum_{k\in F\setminus j}x_k = \alpha.
\]
Therefore,  $(1_j,-1_i)$ is no more a ray.

($\Rightarrow $) Let $x\in \cC(0)$ and satisfy the additional equality $x(F)=0$ for
some $F\in \cF$. Suppose that for some $\alpha>0$, $x':=x+\alpha (1_j,-1_i)$
does not belong to $\cC(v)\cap \{x(F)=0\}$. It means that
\[
\sum_{k\in F}x'_k = \sum_{k\in K}x_k + \alpha\delta_F(j)-\alpha\delta_F(i)\neq 0,
\]
where $\delta_F(k)=1$ if $k\in F$ and 0 otherwise. This implies
$\delta_F(i)\neq\delta_F(j)$, therefore either $i$ or $j$ belongs to $F$, but
not both. Because $j\prec i$ and that a set $F\in \cF$ corresponds to a downset
in $(N,\leq)$, it
must be $j\in F$ and $i\not\in F$.
\end{proof}
\begin{lemma}\label{lem:2}
The minimum number of additional equalities needed to make $\cC(v)$ bounded is $h(N)$.
\end{lemma} 
\begin{proof}
Let us assume that all rays are killed.  By definition of the height, there
exists a maximal chain in $(N,\leq)$ of length $h(N)$ going from a minimal
element to a maximal element, say $i_0,i_1,\ldots,i_{h(N)}$. Then by
Theorem \ref{th:1}, $(1_{i_0},-1_{i_1})$, $(1_{i_1},-1_{i_2})$, \ldots,
$(1_{i_{h(N)-1}},-1_{i_{h(N)}})$ are extremal rays. Because $(1_{i_0},-1_{i_1})$ is
killed, by Lemma~\ref{lem:1} there must be an equality $x(K_1)=0$ such that $i_0\in
K_1$ and $i_1\not\in K_1$. Moreover, since $K_1$ must be a downset,
$i_2,\ldots,i_{h(N)}$ cannot belong to $K_1$. Similarly, there must exist an
equality $x(K_2)=0$ killing ray $(1_{i_1},-1_{i_2})$ such that $i_1\in K_2$ and
$i_2,\ldots,i_{h(N)}\not\in K_2$. Therefore, $K_1\neq K_2$. Continuing this
process we construct a sequence of distinct $h(N)$ subsets
$K_1,K_2,\ldots,K_{h(N)}$, the last one killing ray
$(1_{i_{h(N)-1}},-1_{i_{h(N)}})$. Therefore, at least $h(N)$ equalities are
needed.
\end{proof}
The next algorithm shows an optimal way to define equalities to kill all extremal
rays. It is optimal in the sense that it uses only $h(N)$ equalities and each
equality is the ``smallest'' possible (in the number of terms, or equivalently, in
the size of $F$).

\pagebreak
 
\begin{quote}
{\sc Algo 1}
\begin{description}
\item[Step 0] Initialization. Set $M=N$.
\item[Step 1]  Remove all disconnected elements
  from $M$ (i.e., those elements which are both minimal and maximal). If
  $M=\emptyset$, then STOP. Otherwise, go to Step 2.
\item[Step 2] Build $M_0$ the set of all minimal elements of $M$, and set
  equality $x(\downarrow \! M_0)=0$, where $\downarrow\! M_0$ is the downset
  generated by $M_0$ in $(N,\leq)$.
\item[Step 3] Set $M\leftarrow M\setminus M_0$, and  go to Step 1.
\end{description}
\end{quote}
\begin{theorem}\label{th:2}
{\sc Algo 1} kills all extremal rays and is optimal.
\end{theorem}
\begin{proof}
Steps 1 and 2 build subsets of the level sets of $(N,\leq)$, except the last $h(N)$th level, because
in Step 2, all maximal elements of $N$ are suppressed. Therefore, the algorithm
necessarily finishes in exactly $h(N)$ iterations, and builds $h(N)$
equalities. By Lemma~\ref{lem:2}, this number is optimal.

Consider the first occurrence of Step 2, where $M_0$ is the set of minimal
elements of $N$ (minus those disconnected). Clearly, the equality $x(M_0)=0$ kills all rays of
the form $(1_j,-1_i)$, where $j$ is a minimal element and $i$ is a successor of
$j$ (i.e., $j\prec i$). Therefore, all such $i$'s belong to the level 2. Taking
a proper subset of $M_0$ will necessarily leave some rays of this form, and
subsequent iterations will not kill them. This proves that in each step $M_0$
has a minimal size.

For each iteration, it is not necessary to keep elements $i$ 
which have no successors (i.e., they are maximal), because there cannot exist
rays of the form $(1_i,-1_k)$.  Therefore those elements are suppressed in Step
1. All other elements are necessary since they have a successor and therefore
generate a ray. This proves that in any iteration, $M_0$ has the minimal size,
and so $\downarrow\! M_0$ too. 
\end{proof}

We call the normal collection $\cN$ found by
{\sc Algo 1} the collection of \emph{irredundant normal sets} or \emph{irredundant
  (normal) collection}. We introduce another one, which we call the collection
of \emph{Weber normal sets} or the \emph{Weber (normal) collection} (reasons for
this will be clear after). Supposing $\cN=\{N_1,\ldots,N_{h(N)}\}$ is the
  irredundant collection, the Weber collection is $\{N_1,N_1\cup N_2,N_1\cup
  N_2\cup N_3,\ldots,N_{1}\cup\cdots\cup N_{h(N)} \}$.
\begin{lemma}\label{lem:3}
The Weber collection is a normal collection which is a chain in $\cO(N)$. 
\end{lemma}
\begin{proof}
Lemma~\ref{lem:1} shows that the collection is normal (only elements below those
in the irredundant sets are added). The second assertion is obvious by construction.
\end{proof}
Recall that the normal collection introduced by Grabisch and Xie is
$(L_1,L_1\cup L_2,\ldots,L_1\cup\cdots\cup L_{q-1})$, where $L_1,\ldots,L_q$ are
the level sets of $(N,\leq)$. By construction, $N_1\subseteq L_1$, $N_2\subseteq
L_1\cup L_2$, etc., with proper inclusion in general. This shows that in general
the three normal collections introduced so far differ.

When a normal collection forms a chain, we say that the collection is
\emph{nested}. Note that the Weber collection is the ``smallest'' nested
collection, in the sense that no other nested collection can contain proper
subsets of the Weber collection. Indeed, it is built from the irredundant normal
collection by adding the minimum number of elements to make the collection a
chain. 

Interestingly, the normal collection of Grabisch and Xie is also nested, and it is the
``largest'' nested collection\footnote{Note that this collection is still
  optimal in number of normal sets. ``Largest'' applies here for the size of the
normal sets.}, in the sense that no other nested collection can
contain supersets of this normal collection. Indeed, since a normal set is
built from the union of all level sets up to a given height, adding a new
element $i$ means adding an element from a higher level. Then $(1_k,-1_i)$ for
some $k\prec i$ is an extremal ray, which will not be killed if $i$ is
incorporated into the normal set. Consequently, any nested collection (with
optimal number of normal sets) is comprised between the Weber collection and the
Grabisch-Xie collection.

The following example illustrates that the three normal collections differ.
\begin{example}
Consider the following poset $(N,\leq)$ of 9 elements.
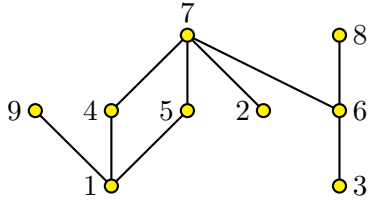
\begin{center}
\psset{unit=0.5cm}
\pspicture(0,-0.5)(8,4.5)
\pspolygon(2,0)(2,2)(4,4)(4,2)
\psline(0,2)(2,0)
\psline(6,2)(4,4)(8,2)
\psline(8,0)(8,4)
\pscircle[fillstyle=solid,fillcolor=yellow](0,2){0.2}
\pscircle[fillstyle=solid,fillcolor=yellow](2,0){0.2}
\pscircle[fillstyle=solid,fillcolor=yellow](2,2){0.2}
\pscircle[fillstyle=solid,fillcolor=yellow](4,2){0.2}
\pscircle[fillstyle=solid,fillcolor=yellow](4,4){0.2}
\pscircle[fillstyle=solid,fillcolor=yellow](6,2){0.2}
\pscircle[fillstyle=solid,fillcolor=yellow](8,0){0.2}
\pscircle[fillstyle=solid,fillcolor=yellow](8,2){0.2}
\pscircle[fillstyle=solid,fillcolor=yellow](8,4){0.2}
\uput[180](0,2){\small 9}
\uput[180](2,0){\small 1}
\uput[180](2,2){\small 4}
\uput[180](4,2){\small 5}
\uput[90](4,4){\small 7}
\uput[180](6,2){\small 2}
\uput[0](8,0){\small 3}
\uput[0](8,2){\small 6}
\uput[0](8,4){\small 8}
\endpspicture
\end{center}
Level 1 is $\{1,2,3\}$, level 2 is $\{4,5,6,9\}$ and level 3 is
$\{7,8\}$. Extremal rays are 
\[
(1_1,-1_9),(1_1,-1_4),(1_1,-1_5),(1_3,-1_6),(1_4,-1_7),(1_5,-1_7),(1_2,-1_7),(1_6,-1_7),(1_6,-1_8). 
\]
The two irredundant normal sets built by {\sc Algo 1} are 123 and 13456, the two
Weber normal sets are 123 and 123456, and the Grabisch-Xie normal sets are 123
and 1234569.
\end{example}

\subsection{The Weber set}\label{sec:weber}
Let us denote by $\eC$ the set of all maximal chains from $\emptyset$ to $N$ in
$\cF$. Consider any maximal chain $C\in \eC$ and its associated
permutation $\sigma$ on $N$, i.e.,
\[
C=\{\emptyset,S_1,S_2,\ldots,S_n=N\},
\] 
with $S_i:=\{\sigma(1),\ldots, \sigma(i)\}$, $i=1,\ldots,n$.  Considering a game
$v$ on $\cF$, the \emph{marginal vector} $x^C\in \bR^N$ associated to $C$ is the
payoff vector defined by
\[
x^C_{\sigma(i)}:= v(S_i) - v(S_{i-1}), \quad i\in N.
\]
The \emph{Weber set} is the convex hull of all marginal vectors:
\[
\cW(v):=\mathrm{conv}(x^C\mid C\in \eC).
\]
In the classical case $\cF=2^N$, it is well known that for any game $v$ it holds
$\cC(v)\subseteq \cW(v)$, with equality if and only if $v$ is convex. In our
general case, this inclusion cannot hold any more since the core is unbounded in
general. We propose a restricted version of the Weber set so that the classical
results still hold.

Consider a nested normal collection (like the Weber collection or the
Grabisch-Xie one) $\cN=\{N_1,\ldots,N_{h(N)}\}$. A \emph{restricted maximal
  chain} (w.r.t. $\cN$) is a maximal chain from $\emptyset$ to $N$ in $\cO(N)$
containing $\cN$. A \emph{restricted marginal vector} is a (classical) marginal vector
whose underlying maximal chain is restricted. The \emph{(restricted) Weber set} $\cW_\cN(v)$
is the convex hull of all restricted marginal vectors w.r.t. $\cN$. The
(unrestricted) Weber set corresponds to the situation $\cN=\emptyset$. 
\begin{lemma}\label{lem:0}
For any restricted maximal chain $C$, its associated restricted marginal vector
$x^C$ coincides with $v$ on $C$, i.e., $x^C(S)=v(S)$ for all $S\in C$.
\end{lemma}
(obvious from the definition)

We recall the following result (see Fujishige and Tomizawa
\cite{futo83,fuj05b}).
\begin{theorem}\label{th:0}
Let $v$ be a game on $\cO(N)$. Then $\cC(v)=\cW(v)$ if and only if $v$ is convex.
\end{theorem}

The following theorems generalize results of \cite{grxi07} and provide more
elegant proofs.  
\begin{theorem}\label{th:3}
Consider $\cN$ a nested normal collection. Then for every game $v$ on $\cO(N)$,
$\cC_\cN(v)\subseteq \cW_\cN(v)$.
\end{theorem}
\begin{proof}
We put $\cN:=\{N_1,\ldots,N_q\}$.  We prove the result by the separation
theorem, proceeding as in \cite{der92}. Suppose there exists
$x\in\cC_\cN(v)\setminus \cW_\cN(v)$. Then it exists $y\in\bR^n$ such that
$\langle w,y\rangle > \langle x,y\rangle$ for all $w\in\cW_\cN(v)$.

Let $\pi$ be a permutation on $N$ such that $y_{\pi(1)}\geq y_{\pi(2)}\geq
\cdots\geq y_{\pi(n)}$. Let us build a permutation $\pi'$ from $\pi$ so that
$\pi'$ corresponds to a restricted maximal chain as follows:
\begin{quote}
Order the elements of $N_1$ according to the $\pi$ order; then order the elements of
$N_2$ according to the $\pi$ order and put them after ; etc. Lastly, put the
remaining elements (in $N\setminus (N_1\cup\cdots\cup N_q )$) according to the
$\pi$ order.
\end{quote}
Note that $\pi'=\pi$ if $\pi$ corresponds to a restricted maximal chain.
With Example 1 and the Weber collection, taking $\pi= 1,4,5,2,9,3,6,7,8$ leads
to $\pi' =1,2,3,4,5,6,9,7,8 $.

Denoting by $m^{\pi'}$ the marginal vector associated to $\pi'$ we have
\begin{align*}
\langle m^{\pi'},y\rangle & = \sum_{i=1}^n
y_{\pi'(i)}\big(v(\{\pi'(1),\ldots,\pi'(i)\}) -
v(\{\pi'(1),\ldots,\pi'(i-1)\})\big)\\
 & = y_{\pi'(n)}v(N) + \sum_{i=1}^{n-1}(y_{\pi'(i)} - y_{\pi'(i+1)})v(\{\pi'(1),\ldots,\pi'(i)\}).
\end{align*}
We claim that if $y_{\pi'(i)} - y_{\pi'(i+1)}<0$ then $\{\pi'(1),\ldots,\pi'(i)\}$
is a normal set. Indeed, by construction of $\pi'$, the situation $y_{\pi'(i)} -
y_{\pi'(i+1)}<0$ can arise only if $\pi'(i)\in N_j$ for some $j$ and
$\pi'(i+1)\in N_{j+1}$. But then by construction again
$N_j=\{\pi'(1),\ldots,\pi'(i)\}$, which proves the claim. 

Therefore since $x\in\cC_\cN(v)$  we have
\begin{align*}
\langle m^{\pi'},y\rangle & \leq y_{\pi'(n)}x(N) + \sum_{i=1}^{n-1}(y_{\pi'(i)}
- y_{\pi'(i+1)})x(\{\pi'(1),\ldots,\pi'(i)\})\\
 & = \sum_{i=1}^ny_{\pi'(i)}x(\{\pi'(1),\ldots,\pi'(i)\}) - \sum_{i=2}^n
y_{\pi'(i)}x(\{\pi'(1),\ldots,\pi'(i-1)\}) \\
 & = \sum_{i=1}^ny_{\pi'(i)}x_{\pi'(i)} = \langle y,x\rangle,
\end{align*}
a contradiction with the assumption.
\end{proof}
\begin{theorem}\label{th:4}
Consider $\cN$ a nested normal collection.
If $v$ is convex on $\cO(N)$, then $\cC_\cN(v)=\cW_\cN(v)$.
\end{theorem}
\begin{proof}
By Theorem~\ref{th:3}, it suffices to show that any restricted marginal vector
is a vertex of $\cC_\cN(v)$. We know already from Theorem~\ref{th:0} that it is
a vertex of $\cC(v)$. It remains to show that any marginal vector satisfies the
normality conditions $x(N_i)=v(N_i)$, $i=1,\ldots,q$, but this is established in
Lemma~\ref{lem:0}. 
\end{proof}

\section{The general case}\label{sec:gene}
We suppose now that $\cF$ is an arbitrary set system. We introduce $\tF$ the
closure of $\cF$ under union and intersection, i.e., the smallest set system
closed under union and intersection containing $\cF$. It is obtained by
iteratively augmenting $\cF$ with unions and intersections of pairs of subsets
of the current set system (starting with $\cF$), till there is no more change in
the set system. As in Section~\ref{sec:cloui}, we assume that $\tF$ has height
$n$ (i.e., it has $n$ join-irreducible elements).
\begin{theorem}\label{th:gene}
Consider an arbitrary set system $\cF$, and assume that its
closure $\tF$ has height $n$. Denote by $\cC(0)$ and $\tC(0)$ the recession
cones generated by $\cF$ and $\tF$. Then $\cC(0)$ and $\tC(0)$ have the same
extremal rays (i.e., $\cC(0)=\tC(0)$) if and only if all extremal rays of
$\cC(0)$ are of the form $(1_j,-1_i)$, for some $i,j\in N$.
\end{theorem}
\begin{proof}
The ``only if'' part is obvious from Theorem~\ref{th:1}. Let us prove the ``if'' part.
Suppose $r$ is an extremal ray of $\cC(0)$. By hypothesis, it has the
form $(1_j,-1_i)$ for some $i,j\in N$. Also, by definition, it satisfies the
system $r(S)\geq 0$ for all $S\in \cF$, which gives $1_S(j)-1_S(i)\geq 0$ for
all $S\in \cF$, which implies that there is no $S\in \cF$ such that $S\ni i$ and
$S\not\ni j$. Therefore it suffices to show that no such $S$ exists in $\tF$. We
show this by induction since $\tF$ is obtained iteratively from $\cF$. We first
prove that the union or intersection of two sets $S_1,S_2$ of $\cF$ cannot at
the same time contain $i$ and not $j$. For intersection, if $S_1\cap S_2\ni i$,
then $S_1,S_2$ too, so they cannot contain $j$, which implies $S_1\cap
S_2\not\ni j$. Now, suppose that $S_1\cup S_2$ does not
contain $j$, which implies that neither $S_1$ nor $S_2$ contain $j$. If
$i\in S_1\cup S_2$, then $i$ belongs at least to one of the sets $S_1,S_2$,
which contradicts the hypothesis. Assume now that the hypothesis holds up to
some step in the iteration process. Clearly, the same reasoning applies again, which
proves that $r$ is a ray of $\tC(0)$. Hence we have proved $\cC(0)\subseteq \tC(0)$.

Conversely, suppose $r$ is an extremal  ray of $\tC(0)$, hence of the form
$(1_j,-1_i)$ by Theorem~\ref{th:1}. Then it satisfies the
system $r(S)\geq 0$ for all $S\in \tF$, and $r(N)=0$. Hence in particular it
satisfies the system $r(S)\geq 0$ for all $S\in \cF$ and $r(N)=0$, and therefore
$r$ is a ray of $\cC(0)$. Therefore $\tC(0)\subseteq \cC(0)$. Hence, we have
proved $\cC(0)=\tC(0)$ and so extremal rays of $\cC(v)$ and $\tC(v)$ are identical.
\end{proof}

Unfortunately, not all set systems $\cF$, even if $\tF$ has height $n$, induce
extremal rays of the form $(1_j,-1_i)$, as shown in the next example. 
\begin{example}\label{ex:2}
Consider $N=\{1,2,3,4\}$, the following set system $\cF$ and its closure $\tF$.
\begin{figure}[htb]
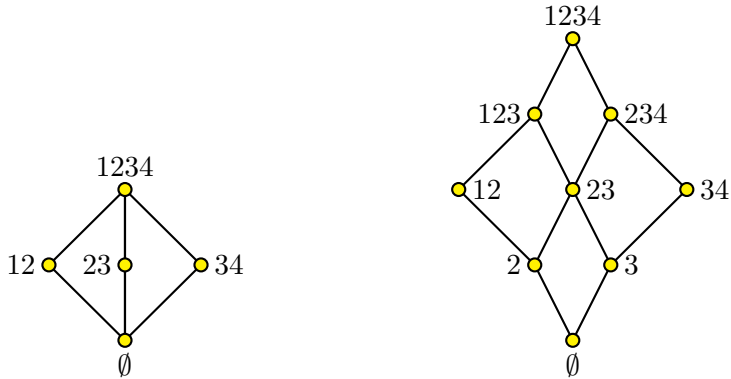

\begin{center}
\psset{unit=0.5cm}
\pspicture(0,-0.5)(4.5,4.5)
\pspolygon(2,0)(0,2)(2,4)
\pspolygon(2,0)(4,2)(2,4)
\pscircle[fillstyle=solid,fillcolor=yellow](0,2){0.2}
\pscircle[fillstyle=solid,fillcolor=yellow](2,2){0.2}
\pscircle[fillstyle=solid,fillcolor=yellow](4,2){0.2}
\pscircle[fillstyle=solid,fillcolor=yellow](2,0){0.2}
\pscircle[fillstyle=solid,fillcolor=yellow](2,4){0.2}
\uput[-90](2,0){\small $\emptyset$}
\uput[180](0,2){\small $12$}
\uput[180](2,2){\small $23$}
\uput[0](4,2){\small $34$}
\uput[90](2,4){\small $1234$}
\endpspicture
\hspace*{3cm}
\psset{unit=0.5cm}
\pspicture(0,-0.5)(6.5,8.5)
\pspolygon(3,0)(2,2)(3,4)(4,2)
\pspolygon(3,4)(2,6)(3,8)(4,6)
\psline(2,2)(0,4)(2,6)
\psline(4,2)(6,4)(4,6)
\pscircle[fillstyle=solid,fillcolor=yellow](3,0){0.2}
\pscircle[fillstyle=solid,fillcolor=yellow](2,2){0.2}
\pscircle[fillstyle=solid,fillcolor=yellow](4,2){0.2}
\pscircle[fillstyle=solid,fillcolor=yellow](0,4){0.2}
\pscircle[fillstyle=solid,fillcolor=yellow](3,4){0.2}
\pscircle[fillstyle=solid,fillcolor=yellow](6,4){0.2}
\pscircle[fillstyle=solid,fillcolor=yellow](2,6){0.2}
\pscircle[fillstyle=solid,fillcolor=yellow](4,6){0.2}
\pscircle[fillstyle=solid,fillcolor=yellow](3,8){0.2}
\uput[-90](3,0){\small $\emptyset$}
\uput[180](2,2){\small $2$}
\uput[0](4,2){\small $3$}
\uput[0](0,4){\small $12$}
\uput[0](3,4){\small $23$}
\uput[0](6,4){\small $34$}
\uput[180](2,6){\small $123$}
\uput[0](4,6){\small $234$}
\uput[90](3,8){\small $1234$}
\endpspicture
\caption{Set system $\cF$ (left) and its closure under union and intersection
  $\tF$ (right)}
\label{fig:1}
\end{center}
\end{figure}
The extremal rays of $\cF$ are $(1,-1,1,-1)$, $(-1,1,-1,1)$ and $(0,0,1,-1)$,
while the extremal rays of $\tF$ are $(-1,1,0,0)$ and $(0,0,1,-1)$. Note that
the first two rays of $\cF$ in fact define a line, and that $\cF$ is neither
regular nor weakly union-closed. 
\end{example} 

Suppose now that $\cF$ has rays of the form $(1_j,-1_i)$. How to kill them?
Lemma~\ref{lem:1} tells us how to kill rays of $\cF$, by considering the
equality $x(F)=0$ with $j\in F$ and $i\not\in F$. Therefore, the only thing we
have to prove is that in any case, such a set $F$ exists in $\cF$. 
\begin{lemma}
Let $\cF$ be a set system such that all extremal rays of $\cC(0)$ are of the
form $(1_j,-1_i)$. Then for each extremal ray $(1_j,-1_i)$, there exists a set
$F\in \cF$ such that $j\in F$ and $i\not\in F$.
\end{lemma} 
\begin{proof}
We consider the ray $(1_j,-1_i)$. We know that in $\tF$ it exists $F_0$ such
that $j\in F_0$ and $i\not\in F_0$. Suppose that no such $F$ exists in $\cF$ and
show that in this case $F_0$ cannot exist in $\tF$. We suppose therefore that
in $\cF$ all sets satisfy either $F\not\ni j$ or $F\ni i$ and we consider two
sets $F_1,F_2$. Observe that we have four possible situations: 1) $F_1\not\ni j$
and $F_2\not\ni j$, 2) $F_1\ni i,j$ and $F_2\ni i,j$ 3) $F_1\not\ni j$ and
$F_2\ni i,j$, and 4) $F_1\ni i,j$ and $F_2\not\ni j$. In all four situations,
we cannot have both $F_1\cup F_2\ni j$ and $F_1\cup F_2\not\ni i$, and the same is
true for $F_1\cap F_2$. Therefore, after one iteration, the
set system has the same property than $\cF$, and so by successive iterations,
$F_0$ cannot be built.
\end{proof}
The above lemma tells us that it is possible to kill rays for such set systems
by turning at most $r$ inequalities to equalities, if $r$ is the number of
rays. Is it possible to give a better answer by using results from
Section~\ref{sec:normal} on $\tF$? Unfortunately, it does not seem possible to
give a general answer here, even for regular set systems. This is because the
irredundant normal sets found by {\sc Algo 1} or the Weber normal collection of
$\tF$ need not belong to $\cF$, as the following simple example shows.
\begin{example}\label{ex:4}
Consider $N=\{1,2,3,4\}$, the following set system $\cF$ (which is regular) and its closure $\tF$.
\begin{figure}[htb]
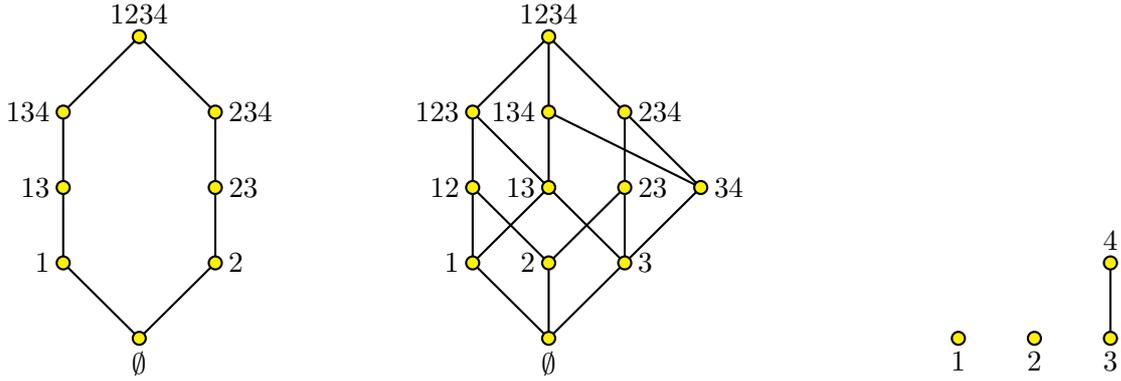

\begin{center}
\psset{unit=0.5cm}
\pspicture(0,-0.5)(4.5,8.5)
\pspolygon(2,0)(0,2)(0,4)(0,6)(2,8)(4,6)(4,4)(4,2)
\pscircle[fillstyle=solid,fillcolor=yellow](0,2){0.2}
\pscircle[fillstyle=solid,fillcolor=yellow](0,4){0.2}
\pscircle[fillstyle=solid,fillcolor=yellow](0,6){0.2}
\pscircle[fillstyle=solid,fillcolor=yellow](2,0){0.2}
\pscircle[fillstyle=solid,fillcolor=yellow](2,8){0.2}
\pscircle[fillstyle=solid,fillcolor=yellow](4,6){0.2}
\pscircle[fillstyle=solid,fillcolor=yellow](4,4){0.2}
\pscircle[fillstyle=solid,fillcolor=yellow](4,2){0.2}
\uput[-90](2,0){\small $\emptyset$}
\uput[180](0,2){\small $1$}
\uput[180](0,4){\small $13$}
\uput[180](0,6){\small $134$}
\uput[90](2,8){\small $1234$}
\uput[0](4,6){\small $234$}
\uput[0](4,4){\small $23$}
\uput[0](4,2){\small $2$}
\endpspicture
\hspace*{3cm}
\psset{unit=0.5cm}
\pspicture(0,-0.5)(6.5,8.5)
\pspolygon(2,0)(0,2)(0,4)(0,6)(2,8)(4,6)(4,4)(4,2)
\psline(2,0)(2,2)(0,4)
\psline(0,2)(2,4)(4,2)(6,4)(4,6)
\psline(2,2)(4,4)
\psline(0,6)(2,4)(2,6)(6,4)
\psline(2,6)(2,8)
\pscircle[fillstyle=solid,fillcolor=yellow](0,2){0.2}
\pscircle[fillstyle=solid,fillcolor=yellow](0,4){0.2}
\pscircle[fillstyle=solid,fillcolor=yellow](0,6){0.2}
\pscircle[fillstyle=solid,fillcolor=yellow](2,0){0.2}
\pscircle[fillstyle=solid,fillcolor=yellow](2,2){0.2}
\pscircle[fillstyle=solid,fillcolor=yellow](2,4){0.2}
\pscircle[fillstyle=solid,fillcolor=yellow](2,6){0.2}
\pscircle[fillstyle=solid,fillcolor=yellow](2,8){0.2}
\pscircle[fillstyle=solid,fillcolor=yellow](4,6){0.2}
\pscircle[fillstyle=solid,fillcolor=yellow](4,4){0.2}
\pscircle[fillstyle=solid,fillcolor=yellow](4,2){0.2}
\pscircle[fillstyle=solid,fillcolor=yellow](6,4){0.2}
\uput[-90](2,0){\small $\emptyset$}
\uput[180](0,2){\small $1$}
\uput[180](0,4){\small $12$}
\uput[180](0,6){\small $123$}
\uput[90](2,8){\small $1234$}
\uput[180](2,6){\small $134$}
\uput[180](2,4){\small $13$}
\uput[180](2,2){\small $2$}
\uput[0](4,6){\small $234$}
\uput[0](4,4){\small $23$}
\uput[0](4,2){\small $3$}
\uput[0](6,4){\small $34$}
\endpspicture
\hspace*{3cm}
\psset{unit=0.5cm}
\pspicture(0,-0.5)(4.5,2.5)
\psline(4,0)(4,2)
\pscircle[fillstyle=solid,fillcolor=yellow](0,0){0.2}
\pscircle[fillstyle=solid,fillcolor=yellow](2,0){0.2}
\pscircle[fillstyle=solid,fillcolor=yellow](4,0){0.2}
\pscircle[fillstyle=solid,fillcolor=yellow](4,2){0.2}
\uput[-90](0,0){\small $1$}
\uput[-90](2,0){\small $2$}
\uput[-90](4,0){\small $3$}
\uput[90](4,2){\small $4$}
\endpspicture
\caption{Set system $\cF$ (left), its closure under union and intersection
  $\tF$ (center), and the generating poset $(N,\leq)$ (right)}
\label{fig:3}
\end{center}
\end{figure}
The unique ray of $\cC(0)$ is (0,0,1,-1). Application of {\sc Algo 1} on $\tF$ gives
as normal set 3 (the Weber normal set is therefore the same). However, 3 does not belong
to $\cF$. Either 13 or 23 can be taken instead. Note that the Grabisch-Xie
normal set is 123, which does not belong either to $\cF$.  
\end{example}
Hence, the only thing which can be done is to build $\tF$, apply {\sc Algo 1} or
compute the Weber normal collection. If some normal sets do not belong to $\cF$,
take the smallest ones of $\cF$ containing them and obeying
Lemma~\ref{lem:1}. It is not guaranted however that we do not need more normal
sets than for $\tF$ (but we do not have an example for this). 

In the rest of the paper, we study two particular types of sets systems, namely
regular set systems and weakly union-closed set systems, which both generalize
systems closed under union and intersection, and where the above results can be applied.

\subsection{The case of regular set systems}
Recall that any maximal chain induces a total order (permutation) on $N$, and
therefore giving a regular set system $\cF$ is equivalent to giving a set of
(permitted) total orders on $N$.
\begin{theorem}\label{th:regu}
Suppose $\cF$ is a regular set system. Then all extremal rays of $\cC(0)$ have
the form $(1_l,-1_m)$ for some $l,m\in N$.
\end{theorem} 
\begin{proof}
Let $\eC$ be the set of all maximal chains from $\emptyset$ to $N$ in $\cF$, and
consider a particular chain, say $\emptyset,\{i\},\{i,j\},\{i,j,k\},\ldots,N$,
inducing the total order $i,j,k,\ldots,$ on $N$, and let us
construct an extremal ray $r$. 

Suppose $r_i>0$, hence w.l.o.g. we can set $r_i=1$. By the condition $r(N)=0$,
there must be at least one $\ell\in N\setminus i$ such that $r_\ell<0$. Select
$\ell$ such that $\ell$ is ranked after $i$ in every maximal chain in
$\eC$. Observe that $(1_i,-1_\ell)$ is a solution of the system $r(S)\geq 0$ for
all $S\in \cF$ and $r(N)=0$ (i.e., it is a ray of $\cC(0)$) if and only if $\ell$ has
the above property, because any $S\ni \ell$ contains also $i$. If no such $\ell$
exists, then set $r_i=0$, which gives a new system of inequalities where $r_i$
has disappeared, and consider the next element $j$ and do the same (note that if
exhausting all elements $i,j,k,\ldots$ without finding $\ell$, is equivalent to
the fact that there is no ray, a situation which happens for example if all
orders exist, i.e., $\cF=2^N$).  Suppose now that there
exist several $\ell$ ranked after $i$ in every maximal chain, say
$\ell_1,\ldots, \ell_q$. Then for every $\alpha_1,\ldots,\alpha_q\geq 0$ such
that $\sum_{p=1}^q \alpha_p=1$, the vector
$(1_i,-\alpha_11_{\ell_1},\ldots,-\alpha_q1_{\ell_q})$ is a ray. But each
$(1_i,-1_{\ell_p})$, $p=1,\ldots,q$ is also a ray, and
$(1_i,-\alpha_11_{\ell_1},\ldots,-\alpha_q1_{\ell_q})$ can be expressed as a
convex combination of these rays, proving that it is not extremal. Therefore
extremal rays are necessarily of the form $(1_i,-1_\ell)$.  In addition, if
$\ell_2$ is ranked after $\ell_1$ in every order, then
$(1_{\ell_1},-1_{\ell_2})$ is a ray, therefore $(1_i,-1_{\ell_2})$ is not
extremal since it can be obtained as
$(1_i,-1_{\ell_1})+(1_{\ell_1},-1_{\ell_2})$ (and similarly for the others).
\end{proof}
By Theorem \ref{th:gene}, we deduce immediately:
\begin{corollary}\label{cor:1}
If $\cF$ is a regular set system, then $\cC(0)=\tC(0)$.
\end{corollary}
We can also deduce Theorem~\ref{th:1} from the above, and therefore derive
an alternative proof of it:
\begin{corollary}\label{cor:2}
If $\cF$ is regular and union and intersection closed, then the extremal rays
are $(1_j,-1_i)$ with $i\in N$ such that $|J_i|>1$ and $j\in J_i$, $j\prec i$.  
\end{corollary}
\begin{proof}
Under the hypothesis, $\cF$ is generated by a poset $(N,\leq)$, and the set of
total orders generated by the maximal chains are those orders compatible with
the partial order $\leq$ on $N$. Then it is easy to see from the proof of
Theorem~\ref{th:regu} that we obtain the desired extremal rays.  
\end{proof}
The proof of Theorem \ref{th:regu} being constructive, we can propose the
following simple algorithm to produce all extremal rays of a regular set system.
\begin{quote}
{\sc Algo 2}
\begin{description}
\item[Step 0] Initialization. Select a maximal chain $C$ in $\eC$, and denote
  for simplicity by $1,2,\ldots,n$ the order induced by $C$. Put $L=\emptyset$. 
\item[For ] $i=1$ to $n-1$ {\bf do}:
  \begin{description}
  \item[For ] $j=i+1$ to $n$ {\bf do}:\\
    {If } $j$ is ranked after $i$ in every chain in $\eC$, {\bf then}
    \begin{itemize}
    \item Put $(1_i,-1_j)$ in $L$ \\ \emph{\% this is a candidate for being an
      extremal ray}
    \item {\bf For } $k<i$, check if $(1_k,-1_i)$ and $(1_k,-1_j)$ both
      exist in $L$. {\bf If } yes, remove $(1_k,-1_j)$ from $L$ \\ \emph{\% it
        can be obtained as the sum of $(1_k,-1_i)$ and $(1_i,-1_j)$}
    \end{itemize}
  \end{description}
\item[Final step: ] output list $L$ of extremal rays.
\end{description}
\end{quote}
\begin{example}
Let us apply  {\sc Algo 2} on the regular set system of Fig.~\ref{fig:5}
(left). The four orders induced by the maximal chains are:
\begin{gather*}
1-4-2-3-5\\
2-4-1-3-5\\
2-4-3-5-1\\
2-4-3-1-5
\end{gather*}
Let us take the first order for running the algorithm. Taking $i=1$, we see that
no $j$ can be found. Therefore, we take $i=4$, then $j=3$ and 5 are possible, so
we put in $L$ the rays $(0,0,-1,1,0)$ and $(0,0,0,1,-1)$. Let us take now $i=2$,
then $j=3$ and $5$ are possible, so we add in $L$ the two rays $(0,1,-1,0,0)$
and $(0,1,0,0-1)$. Next, we take $i=3$ and see that $j=5$ is possible, therefore
we put $(0,0,1,0,-1)$ in $L$. However, we have to remove $(0,0,0,1,0,-1)$ and
$(0,1,0,0,-1)$ from $L$. The extremal rays are therefore $(0,0,-1,1,0)$,
$(0,1,-1,0,0)$ and $(0,0,1,0,-1)$. This result is confirmed by the PORTA software.
\end{example}

We end this section by addressing the definition of the Weber set. Since $\cF$
is regular, marginal vectors can be defined as usual and therefore it makes
sense to speak of the Weber set. Suppose we have found a normal nested
collection of sets $\cN$, then the restricted Weber set $\cW_\cN(v)$ for $v$ defined
on $\cF$ can be defined as before. The question is then to compare $\cW_\cN(v)$
with $\cC_\cN(v)$ and also $\tW_{\cN'}(v)$, the restricted Weber set on $\tF$,
with $\cN'$ the Weber normal collection of $\tF$. Little can be said in general
if one does not have $\cN'=\cN$. Suppose then that this is the case. Because of
regularity, any restricted maximal chain in $\cF$ is a restricted maximal chain
in $\tF$, so that we have $\cW_\cN(v)\subseteq \tW_\cN(v)$. Recall also that
$\cC_\cN(v)\supseteq \tC_\cN(v)$, hence the question whether $\cC_\cN(v)\subseteq
\cW_\cN(v)$ remains. An examination of the proof of Theorem~\ref{th:3} reveals that
the technique of the proof cannot extend to this case. Indeed, the following
example shows that this is not true in general.
\begin{example}\label{ex:5}
Consider $N=\{1,2,3,4,5\}$, the following regular set system $\cF$  and its closure $\tF$.
\begin{figure}[htb]
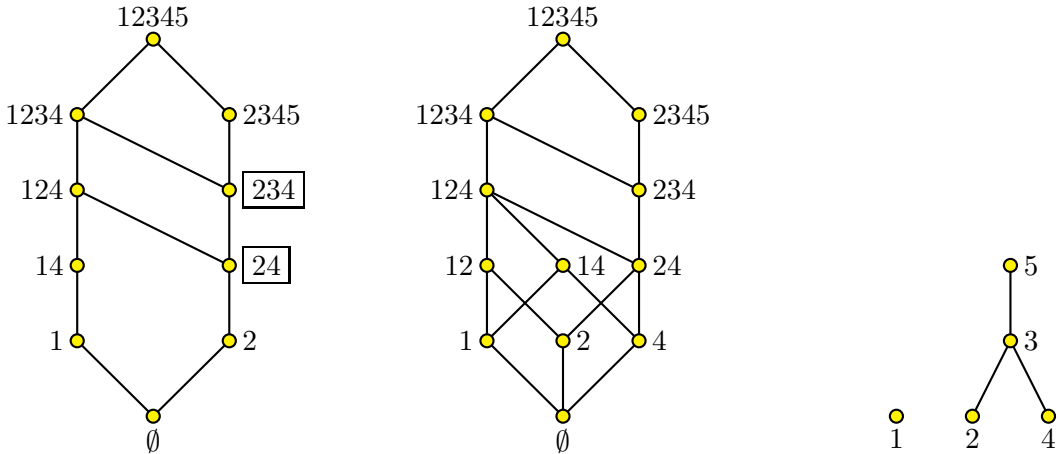

\begin{center}
\psset{unit=0.5cm}
\pspicture(0,-0.5)(4.5,10.5)
\pspolygon(2,0)(0,2)(0,8)(2,10)(4,8)(4,2)
\psline(4,4)(0,6)
\psline(4,6)(0,8)
\pscircle[fillstyle=solid,fillcolor=yellow](0,2){0.2}
\pscircle[fillstyle=solid,fillcolor=yellow](0,4){0.2}
\pscircle[fillstyle=solid,fillcolor=yellow](0,6){0.2}
\pscircle[fillstyle=solid,fillcolor=yellow](0,8){0.2}
\pscircle[fillstyle=solid,fillcolor=yellow](2,0){0.2}
\pscircle[fillstyle=solid,fillcolor=yellow](2,10){0.2}
\pscircle[fillstyle=solid,fillcolor=yellow](4,8){0.2}
\pscircle[fillstyle=solid,fillcolor=yellow](4,6){0.2}
\pscircle[fillstyle=solid,fillcolor=yellow](4,4){0.2}
\pscircle[fillstyle=solid,fillcolor=yellow](4,2){0.2}
\uput[-90](2,0){\small $\emptyset$}
\uput[180](0,2){\small $1$}
\uput[180](0,4){\small $14$}
\uput[180](0,6){\small $124$}
\uput[180](0,8){\small $1234$}
\uput[90](2,10){\small $12345$}
\uput[0](4,8){\small $2345$}
\uput[0](4,6){\small $\fbox{234}$}
\uput[0](4,4){\small $\fbox{24}$}
\uput[0](4,2){\small $2$}
\endpspicture
\hspace*{3cm}
\psset{unit=0.5cm}
\pspicture(0,-0.5)(4.5,10.5)
\pspolygon(2,0)(0,2)(0,8)(2,10)(4,8)(4,2)
\psline(2,0)(2,2)(0,4)
\psline(2,2)(4,4)(0,6)(2,4)(0,2)
\psline(2,4)(4,2)  
\psline(4,6)(0,8)
\pscircle[fillstyle=solid,fillcolor=yellow](0,2){0.2}
\pscircle[fillstyle=solid,fillcolor=yellow](0,4){0.2}
\pscircle[fillstyle=solid,fillcolor=yellow](0,6){0.2}
\pscircle[fillstyle=solid,fillcolor=yellow](0,8){0.2}
\pscircle[fillstyle=solid,fillcolor=yellow](2,0){0.2}
\pscircle[fillstyle=solid,fillcolor=yellow](2,2){0.2}
\pscircle[fillstyle=solid,fillcolor=yellow](2,4){0.2}
\pscircle[fillstyle=solid,fillcolor=yellow](2,10){0.2}
\pscircle[fillstyle=solid,fillcolor=yellow](4,8){0.2}
\pscircle[fillstyle=solid,fillcolor=yellow](4,6){0.2}
\pscircle[fillstyle=solid,fillcolor=yellow](4,4){0.2}
\pscircle[fillstyle=solid,fillcolor=yellow](4,2){0.2}
\uput[-90](2,0){\small $\emptyset$}
\uput[1800](2,2){\small $2$}
\uput[0](2,4){\small $14$}
\uput[180](0,2){\small $1$}
\uput[180](0,4){\small $12$}
\uput[180](0,6){\small $124$}
\uput[180](0,8){\small $1234$}
\uput[90](2,10){\small $12345$}
\uput[0](4,8){\small $2345$}
\uput[0](4,6){\small $234$}
\uput[0](4,4){\small $24$}
\uput[0](4,2){\small $4$}
\endpspicture
\hspace*{3cm}
\psset{unit=0.5cm}
\pspicture(0,-0.5)(4.5,4,5)
\psline(2,0)(3,2)(4,0)
\psline(3,2)(3,4)
\pscircle[fillstyle=solid,fillcolor=yellow](0,0){0.2}
\pscircle[fillstyle=solid,fillcolor=yellow](2,0){0.2}
\pscircle[fillstyle=solid,fillcolor=yellow](4,0){0.2}
\pscircle[fillstyle=solid,fillcolor=yellow](3,2){0.2}
\pscircle[fillstyle=solid,fillcolor=yellow](3,4){0.2}
\uput[-90](0,0){\small $1$}
\uput[-90](2,0){\small $2$}
\uput[-90](4,0){\small $4$}
\uput[0](3,2){\small $3$}
\uput[0](3,4){\small $5$}
\endpspicture
\caption{Set system $\cF$ (left), its closure under union and intersection
  $\tF$ (center), and the generating poset $(N,\leq)$ (right)}
\label{fig:5}
\end{center}
\end{figure}
{\sc Algo 1} applied on $\tF$ gives 24 and 234 as normal sets, which is also the Weber
collection. These sets belong to $\cF$, therefore the restricted Weber set can
be defined with the Weber collection. There are only two restricted maximal
chains on $\cF$, namely $\emptyset,2,24,234,2345,N$ and
$\emptyset,2,24,234,1234,N$, inducing the two vertices of $\cW_\cN(v)$:
\begin{align*}
w_1 & =(v(N)-v(2345), v(2), v(234)-v(24), v(24)-v(2), v(2345)-v(234))\\
w_2 & =(v(1234)-v(234), v(2), v(234)-v(24), v(24)-v(2), v(N)-v(1234)).
\end{align*} 
The restricted core is defined by the system:
\begin{align*}
x_1 & \geq v(1)\\
x_2 & \geq v(2)\\
x_1 + x_4 & \geq v(14)\\
x_2 + x_4 & = v(24)\\
x_1 + x_2 + x_4 & \geq v(124)\\
x_2 + x_3 + x_4 & = v(234)\\
x_1 + x_2 + x_3 + x_4 & \geq v(1234)\\
x_2 + x_3 + x_4 + x_5 & \geq v(2345)\\
x_1 + x_2 + x_3 + x_4 + x_5 & = v(N)
\end{align*}
Let us take the game defined by $v(N)=3$, $v(1234)=v(2345)=2$, $v(234)=1$,
$v(124)=2$, $v(24)=v(14)=1$, $v(2)=v(1)=0$. Then the two vertices of the Weber
set are $(1,0,0,1,1)$ and $(1,0,0,1,1)$, which makes the Weber set a
singleton. However, the vector $(1,1,0,0,1)$ is an element of the restricted
core, which forbids the core to be included into the Weber set.
\end{example}

\subsection{The case of weakly union-closed systems}
The situation here is less simple than with regular set systems. The following
theorem gives a sufficient condition for the equality of $\cC(0)$ and $\tC(0)$.
\begin{theorem}
Assume that $\cF$ is a weakly union-closed system, and denote by $\tF$ its
closure under union and intersection. Then the extremal rays of $\cC(0)$ and
$\tC(0)$ are the same if for any $S\in\tF\setminus \cF$, it is either a union of
disjoint sets of $\cF$, or there exist $S_1,S_2\in\cF$
such that $S=S_1\cap S_2$, and there
exists a covering in $\cF$ of $N\setminus (S_1\cup S_2)$.
\end{theorem} 
By definition of weakly union-closed systems, note that the covering will be in
fact a partition.
\begin{proof}
We consider the set of inequalities of $\cC(0)$, i.e., $x(S)\geq 0$ for all
$S\in \cF$ and $x(N)=0$. We will prove that any additional inequality $x(F)\geq
0$ with $F\in \tF\setminus \cF$ is redundant. By the Farkas lemma, we know that
this amounts to prove that $x(F)\geq 0$ can be obtained by a positive linear
combination of the inequalities $x(S)\geq 0$, $S\in \cF$ and $x(N)=0$.

We consider $S\in \tF\setminus \cF$. Assume first that $S$ is a disjoint union
of sets in $\cF$, say $S=S_1\cup\cdots\cup S_k$. Then obviously $x(S)\geq 0$ is
implied by equalities $x(S_i)\geq 0$, $i=1,\ldots,q$, since it can be obtained as their sum.
Suppose on the contrary that $S$ is not a disjoint union of sets in $\cF$.  By
hypothesis, there exists $S_1,S_2\in\cF$ such that $S_1\cap S_2=S$
and there exists a partition $\{T_1,\ldots,T_k\}$ of $N\setminus (S_1\cup S_2)$.
Let us write the following system of inequalities:
\begin{align*}
x(S_1) & \geq 0 && (a_1)\\
x(S_2) & \geq 0 && (a_2)\\
x(T_1) & \geq 0 && (b_1)\\
\vdots & \vdots && \vdots\\
x(T_k) & \geq 0 && (b_{k})\\
-x(N)  &\geq 0 && (c),
\end{align*}
the last one coming from $x(N)=0$. Then the inequality $x(S)\geq 0$ is
obtained by $(a_1) + (a_2) + (b_1) + \cdots (b_k) + (c)$, which proves that
$x(S)\geq 0$ is redundant.   
\end{proof}
The next example illustrates the case where this condition is not satisfied.
\begin{example}
Take $N=\{1,2,3,4\}$ and the following weakly union-closed set system $\cF$ and its closure $\tF$.
\begin{figure}[htb]
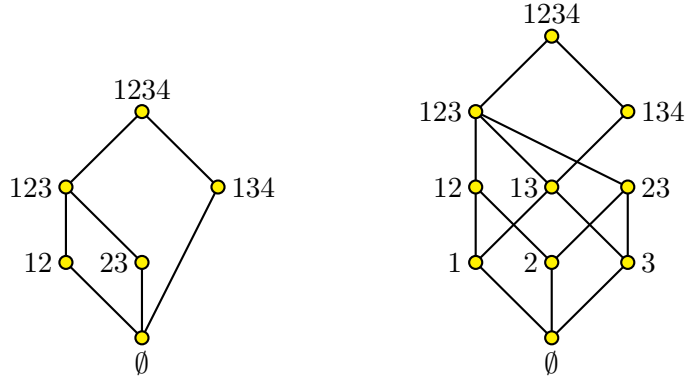

\begin{center}
\psset{unit=0.5cm}
\pspicture(0,-0.5)(4.5,6.5)
\pspolygon(2,0)(0,2)(0,4)(2,2)
\psline(2,0)(4,4)(2,6)
\psline(0,4)(2,6)
\pscircle[fillstyle=solid,fillcolor=yellow](2,0){0.2}
\pscircle[fillstyle=solid,fillcolor=yellow](0,2){0.2}
\pscircle[fillstyle=solid,fillcolor=yellow](2,2){0.2}
\pscircle[fillstyle=solid,fillcolor=yellow](4,4){0.2}
\pscircle[fillstyle=solid,fillcolor=yellow](0,4){0.2}
\pscircle[fillstyle=solid,fillcolor=yellow](2,6){0.2}
\uput[-90](2,0){\small $\emptyset$}
\uput[180](0,2){\small $12$}
\uput[180](2,2){\small $23$}
\uput[180](0,4){\small $123$}
\uput[0](4,4){\small $134$}
\uput[90](2,6){\small $1234$}
\endpspicture
\hspace*{3cm}
\psset{unit=0.5cm}
\pspicture(0,-0.5)(4.5,8.5)
\pspolygon(2,0)(0,2)(0,6)(2,8)(4,6)(2,4)(4,2)
\psline(2,0)(2,2)(0,4)
\psline(0,2)(2,4)(0,6)
\psline(2,2)(4,4)(0,6)
\psline(4,2)(4,4)
\pscircle[fillstyle=solid,fillcolor=yellow](2,0){0.2}
\pscircle[fillstyle=solid,fillcolor=yellow](0,2){0.2}
\pscircle[fillstyle=solid,fillcolor=yellow](2,2){0.2}
\pscircle[fillstyle=solid,fillcolor=yellow](4,2){0.2}
\pscircle[fillstyle=solid,fillcolor=yellow](0,4){0.2}
\pscircle[fillstyle=solid,fillcolor=yellow](2,4){0.2}
\pscircle[fillstyle=solid,fillcolor=yellow](4,4){0.2}
\pscircle[fillstyle=solid,fillcolor=yellow](0,6){0.2}
\pscircle[fillstyle=solid,fillcolor=yellow](4,6){0.2}
\pscircle[fillstyle=solid,fillcolor=yellow](2,8){0.2}
\uput[-90](2,0){\small $\emptyset$}
\uput[180](0,2){\small $1$}
\uput[180](2,2){\small $2$}
\uput[0](4,2){\small $3$}
\uput[180](0,4){\small $12$}
\uput[180](2,4){\small $13$}
\uput[0](4,4){\small $23$}
\uput[180](0,6){\small $123$}
\uput[0](4,6){\small $134$}
\uput[90](2,8){\small $1234$}
\endpspicture
\caption{Set system $\cF$ (left) and its closure under union and intersection
  $\tF$ (right)}
\label{fig:2}
\end{center}
\end{figure}
The required condition fails: take $S=2$, then it can obtained only by the
intersection of 12 and 23. But $N\setminus 123=4$ is not a subset of $\cF$. The extremal
rays of $\cC(0)$ are $(0,0,1,-1)$, $(1,0,0,-1)$ and $(1,-1,1,-1)$, but
$\tC(0)$ has only the two first rays as extremal rays.
\end{example}

\bibliographystyle{plain}
\bibliography{../BIB/fuzzy,../BIB/grabisch,../BIB/general}

\end{document}